\RequirePackage{etoolbox}
\csdef{input@path}{%
 {sty/}
 {img/}
}%
\csgdef{bibdir}{bib/}

\documentclass[ba]{imsart}
\pubyear{0000}
\volume{00}
\issue{0}
\doi{0000}
\firstpage{1}
\lastpage{1}

\usepackage{amsmath,amsfonts,amsthm,amssymb,bm}
\usepackage{natbib}
\usepackage[colorlinks,citecolor=blue,urlcolor=blue,filecolor=blue,backref=page]{hyperref}
\usepackage{graphicx}

\usepackage{algorithm}
\usepackage[noend]{algpseudocode}

\algnewcommand\algorithmicinput{\textbf{Input:}}
\algnewcommand\INPUT{\item[\algorithmicinput]}

\newcommand{\stb}{\State \;}

\startlocaldefs
\numberwithin{equation}{section}
\theoremstyle{plain}

\newtheorem{theorem}{Theorem}[section]

\newtheorem{proposition}{Proposition}[section]

\newtheorem{assumption}{Assumption}[section]

\newtheorem{remark}{Remark}[section]

\endlocaldefs

\begin{document}

\begin{frontmatter}
\title{Nested adaptation of MCMC algorithms}
\runtitle{Nested adaptation MCMC}

\begin{aug}
\author{\fnms{Dao} \snm{Nguyen}\thanksref{addr1,m1}\ead[label=e1]{dxnguyen@olemiss.edu}},
\author{\fnms{Perry} \snm{de Valpine}\thanksref{addr2}},
\author{\fnms{Yves} \snm{Atchade}\thanksref{addr3}},\\
\author{\fnms{Daniel} \snm{Turek}\thanksref{addr4}},
\author{\fnms{Nicholas} \snm{Michaud}\thanksref{addr2,addr5}},
\author{\fnms{Christopher} \snm{Paciorek}\thanksref{addr5}}

\runauthor{Nguyen et al.}

\address[addr1]{Departments of Mathematics, University of Mississippi, Oxford
    \printead{e1} 
}

\address[addr2]{Department of Environmental Science,
Policy, and Management
University of California, Berkeley
}
\address[addr3]{Department of Statistics, 
University of Michigan, 
Ann Arbor
}
\address[addr4]{Department of Mathematics \& Statistics, 
Williams College, 
Williamstown
}
\address[addr5]{Department of Statistics, 
University of California, 
Berkeley
}


\end{aug}

\begin{abstract}
Markov chain Monte Carlo (MCMC) methods are ubiquitous tools for simulation-based inference in many fields but designing and identifying good MCMC samplers is still an open question. 
This paper introduces a novel MCMC algorithm, namely, Nested Adaptation MCMC. 
For sampling variables or blocks of variables, we use two levels of adaptation where the inner adaptation optimizes the MCMC performance within each sampler, 
while the outer adaptation explores the space of valid kernels to find the optimal samplers. 
We provide a theoretical foundation for our approach. To show the generality and usefulness of the approach, 
we describe a framework using only standard MCMC samplers as candidate samplers and some adaptation schemes for both inner and outer iterations. In several benchmark problems, we show that our proposed approach substantially outperforms other approaches, including an automatic blocking algorithm, 
in terms of MCMC efficiency and computational time.
\end{abstract}


\end{frontmatter}

\section{Introduction}

Markov chain Monte Carlo (MCMC) has become a widely used approach
for simulation from an arbitrary distribution of interest, typically
a Bayesian posterior distribution, known as the target distribution.
MCMC really represents a family of sampling methods. Generally speaking,
any new sampler that can be shown to preserve the ergodicity of the
Markov chain such that it converges to the target distribution is
a member of the family and can be combined with other samplers as
part of a valid MCMC kernel. The key to MCMC's success is its
simplicity and applicability. In practice, however, it sometimes needs
a lot of non-trivial tuning to work well \citep{haario2005componentwise}.

To deal with this problem, many adaptive MCMC algorithms have been
proposed \citep{gilks1998adaptive,haario2001adaptive,andrieu2001controlled,sahu2003self}.
These allow parameters of the MCMC kernel to be automatically
tuned based on previous samples. This breaks the Markovian property
of the chain so has required special schemes and proofs that the resulting
chain will converge to the target distribution \citep{andrieu2003ergodicity,atchade2005adaptive,andrieu2006efficiency}.
Under some weaker and easily verifiable conditions, namely ``diminishing
adaptation'' and \textit{``}containment'', \citet{roberts2007coupling}
proved ergodicity of adaptive MCMC and proposed many useful samplers.

It is important to realize, however, that such adaptive MCMC samplers
address only a small aspect of a much larger problem. A typical adaptive MCMC sampler will approximately optimize performance
given the kind of sampler chosen in the first place, but it will not
optimize among the variety of samplers that could have been chosen.
For example, an adaptive random walk Metropolis-Hastings sampler will
adapt the scale of its proposal distribution, but that adaptation
won't reveal whether an altogether different kind of sampler would
be more efficient. In many cases it would, and the exploration of
different sampling strategies often remains a human-driven trial-and-error
affair.

Here we present a method for a higher level of MCMC adaptation. The
adaptation explores a potentially large space of valid MCMC kernels
composed of different samplers. One starts with an arbitrary set of
candidate samplers for each dimension or block of dimensions in the
target distribution. The main idea is to iteratively try different
candidates that compose a valid MCMC kernel, run them for a relatively
short time, generate the next set of candidates based on the results
thus far, and so on. Since relative performance of different samplers
is specific to each model and even to each computing environment,
it is doubtful whether there is a universally optimal kind of sampler.
Hence we view the choice of efficient samplers for a particular problem
as well-suited to empirical determination via computation.

The goal of computationally exploring valid sampler combinations in
search of an efficient model-specific MCMC kernel raises a number
of challenges. First, one must prove that the samples collected as
the algorithm proceeds indeed converge to the target distribution,
even when some of the candidate samplers are internally adaptive,
such as conventional adaptive Metropolis-Hastings samplers. We provide such a proof for a general framework.

Second, one must determine efficient methods for exploring the very
large, discrete space of valid sampler combinations. This is complicated
by a combinatorial explosion, which is exacerbated by the fact that
any multivariate samplers can potentially be used for arbitrary blocks
of model dimensions. Here we take a practical approach to this problem,
setting as our goal only to show basic schemes that can yield substantial
improvements in useful time frames. Future work can aim to develop improvements within
the general framework presented here. We also limit ourselves to relatively simple candidate samplers, 
but the framework can accommodate many more choices.

Third, one must determine how to measure the efficiency of a particular
MCMC kernel for each dimension and for the entire model, in order to
have a metric to seek to optimize. As a first step, it is vital to
realize that there can be a trade-off between good mixing and computational
speed. When considering adaptation within one kind of sampler, say
adaptive Metropolis-Hastings, one can roughly assume that computational cost
does not depend on the proposal scale, and hence mixing measured by
integrated autocorrelation time, or the related effective sample size,
is a sensible measure of efficiency. But when comparing two samplers
with very different computational costs, say adaptive Metropolis-Hastings
and slice samplers, good mixing may or may not be worth its computational
cost. Metropolis-Hastings samplers may mix more slowly than slice samplers
on a per iteration basis, but they do so at higher computational speed
because slice samplers can require many evaluations of model density functions. Thus the greater number of random walk iterations per
unit time could outperform the slice sampler. An additional issue
is that different dimensions of the model may mix at different rates,
and often the slowest-mixing dimensions limit the validity of all
results \citep{turek2017automated}. In view of these considerations, we define MCMC efficiency
as the effective sample size per computation time and use that as
a metric of performance per dimension. Performance of an MCMC kernel
across all dimensions is defined as the minimum efficiency among all
dimensions.

The rest of the paper is organized as follows. Section \ref{theory} begins with a general theoretical framework
for Nested Adaptation MCMC, then extends these methods to a specific Nested Adaptation algorithm
involving block MCMC updating. Section \ref{method} presents an example algorithm
that fits within the framework, and provides some explanations on its
details. Section \ref{example} then outlines some numerical examples comparing
the example algorithm with existing algorithms for a variety of benchmark
models. Finally, section \ref{discussion} concludes and discusses some future research
directions.

\section{A general Nested Adaptation MCMC}\label{theory}

In this section, we present a general Nested Adaptation MCMC algorithm and 
give theoretical results establishing its correctness.

Let $\mathbf{\mathcal{X}}$ be a state space and $\pi$ the probability
distribution on $\mathbf{\mathcal{X}}$ that we wish to sample from.
Let $\mathbf{\mathcal{I}}$ be a countable set (this set indexes the
discrete set of MCMC kernels we wish to choose from). For $\iota\in\mathbf{\mathcal{I}}$,
let $\Theta_{\iota}$ be a parameter space in some subset space $\mathbb{R}^{m}$.
For $\iota\in\mathbf{\mathcal{I}}$ and $\theta\in\Theta_{\iota}$,
let $P_{\iota,\theta}$ denote a Markov kernel on $\mathbf{\mathcal{X}}$
with invariant distribution $\pi$. We set ${\displaystyle \bar{\Theta}=\bigcup_{\iota\in\mathbf{\mathcal{I}}}\{\iota\}\times\Theta_{\iota}}$
the adaptive MCMC parameter space. We want to build a stochastic process
(an adaptive Markov chain) $\{(X_{n},\ \iota_{n},\ \theta_{n}),\ n\ \geq0\}$
on $\mathbf{\mathcal{X}}\times\bar{\Theta}$ such that as $n\rightarrow\infty$,
the distribution of $X_{n}$ converges to $\pi$, and a law of large
numbers holds. We call $\iota$ the external adaptation parameter
and $\theta$ the internal adaptation parameter.

We will follow the general adaptive MCMC recipe of \citet{roberts2009examples}.
Assume that any internal adaptation on $\Theta_{\iota}$ is
done using a function $H_{\iota}$ : $\Theta_{\iota}\times\mathbf{\mathcal{X}}\rightarrow\Theta_{\iota}$,
and an ``internal clock'' sequence $\{\gamma_{n},\ n\ \geq0\}$
such that $\lim_{n\rightarrow\infty}\gamma_{n}=0$. The function $H_{\iota}$ depends on $\gamma_n$ and updates parameters 
for internal adaptation. 
Also, let $\{p_{k},\ k\geq1\}$
be a sequence of numbers $p_{k}\in(0,1)$ such that $\lim_{k\rightarrow\infty}p_{k}=0$.
$p_{k}$ will be the probability of performing external adaptation
at external iteration $k$. During the algorithm
we will also keep track of two variables: $\kappa_{n}$, the number
of external adaptations performed up to step $n$; and $\tau_{n}$,
the number of iterations between $n$ and the last time an external
adaptation is performed. These two variables are used to manage the 
internal clock based on external iterations, which in most situations
can simply be the number of adaptation steps. We build the stochastic
process $\{(X_{n},\ \iota_{n},\ \theta_{n}),\ n\ \geq0\}$ on $\mathbf{\mathcal{X}}\times\bar{\Theta}$
as follows. 
\begin{enumerate}
\item We start with $\kappa_{0}=\tau_{0}=0$. We start also with some $X_{0}\in\mathbf{\mathcal{X}},\ \iota_{0}\in\mathbf{\mathcal{I}}$,
and $\theta_{0}\in\Theta_{\iota_{0}}$. 
\item At the n-th iteration, given $\mathcal{F}_{n}\overset{def}{=}\sigma\{(X_{k},\ \iota_{k},\ \theta_{k}),\ k\leq n\}$,
and given $\kappa_{n},\ \tau_{n}$: 
\begin{enumerate}
\item Draw $X_{n+1}\sim P_{\iota_{n},\theta_{n}}(X_{n},\ \cdot)$. 
\item Independently of $\mathcal{F}_{n}$ and $X_{n+1}$, draw $B_{n+1}\sim$
Bern$(p_{n+1})\in\{0,1\}$. 
\begin{enumerate}
\item If $B_{n+1}=0$, there is no external adaptation: $\iota_{n+1}=\iota_{n}$.
We update $\kappa_{n}$ and $\tau_{n}$: 
\begin{equation}
\kappa_{n+1}=\kappa_{n},\ \tau_{n+1}=\tau_{n}+1.\label{eq:1}
\end{equation}
Then we perform an internal adaptation: set $c_{n+1}=\kappa_{n+1}+\tau_{n+1}$,
and compute 
\begin{equation}
\theta_{n+1}=\theta_{n}+\gamma_{c_{n+1}}H_{\iota_{n}}(\theta_{n},\ X_{n+1}).\label{eq:2}
\end{equation}
Note that the internal adaptation interval could vary between iterations. 
\item If $B_{n+1}=1$, then we do an external adaptation: we choose a new
$\iota_{n+1}$. And we choose a new value $\theta_{n+1}\in\Theta_{\iota_{n+1}}$
based on $\mathcal{F}_{n}$ and $X_{n+1}$. Then we update $\kappa_{n}$
and $\tau_{n}$. 
\begin{equation}
\kappa_{n+1}=\kappa_{n}+1,\ \tau_{n+1}=0.
\end{equation}
\end{enumerate}
\end{enumerate}
\end{enumerate}
For this Nested Adaptation MCMC algorithm to be valid we must show that it
satisfies three assumptions: 
\begin{enumerate}
\item For each $(\iota,\ \theta)\in\bar{\Theta},\ P_{\iota,\theta}$ has
invariant distribution $\pi$. 
\item (diminishing adaptation): 
\[
\triangle_{n+1}\overset{_{def}}{=}\sup_{x\in\mathcal{X}}\Vert P_{\iota_{n},\theta_{n}}(x,\ \cdot)-P_{\iota_{n+1},\theta_{n+1}}(x,\ \cdot)\Vert_{\mathrm{TV}}
\]
converges in probability to zero, as $n\rightarrow\infty$, 
\item (containment): For all $\epsilon>0$, the sequence $\{M_{\epsilon}(\iota_{n},\ \theta_{n},\ X_{n})\}$
is bounded in probability, where 
\[
M_{\epsilon}(\iota,\ \theta,\ x)\overset{_{def}}{=}\inf\{n\ \geq1:\Vert P_{\iota,\theta}^{n}(x,\ \cdot)-\pi\Vert_{\mathrm{TV}}\leq\epsilon\}.
\]
\end{enumerate}
\begin{remark} Here the first assumption holds by construction. We
will show that by the design, our Nested Adaptation algorithm satisfies
the diminishing adaptation.\end{remark}
For $\iota\in\mathbf{\mathcal{I}},\ \theta,\ \theta'\in\Theta_{\iota}$,
define 
\[
D_{\iota}(\theta,\ \theta')\ \overset{_{def}}{=}\ \sup_{x\in\mathbf{\mathcal{X}}}\Vert P_{\iota,\theta}(x,\ \cdot)-P_{\iota,\theta'}(x,\ \cdot)\Vert_{\mathrm{TV}}.
\]
\begin{proposition} Suppose that $\mathbf{\mathcal{I}}$ is finite,
and for any $\iota\in\mathcal{\mathbf{\mathcal{I}}}$, the adaptation
function $H_{\iota}$ is bounded, and there exists $C<\infty$ such
that 
\[
D_{\iota}(\theta,\ \theta')\leq C\Vert\theta-\theta'\Vert.
\]
Then the diminishing adaptation holds. \end{proposition} \begin{proof}
We have 
\begin{eqnarray*}
\mathrm{E}(\triangle_{n+1}) & = & p_{n+1}\mathrm{E}(\triangle_{n+1}|B_{n+1}=1)+(1-p_{n+1})\mathrm{E}(\triangle_{n+1}|B_{n+1}=0),\\
 & \leq & 2p_{n+1}+\mathrm{E}(\triangle_{n+1}|B_{n+1}=0),\\
 & = & 2p_{n+1}+\mathrm{E}\ [D_{\iota_{n}}(\theta_{n},\ \theta_{n+1})],\\
 & \leq & 2p_{n+1}+C\mathrm{E}\ [\Vert\theta_{n+1}-\theta_{n}\Vert],\\
 & \leq & 2p_{n+1}+C_{1}\gamma_{c_{n+1}},
\end{eqnarray*}
where $c_{n+1}=\kappa_{n+1}+\tau_{n+1}$. It is easy to see that $c_{n}\rightarrow\infty$
as $n\rightarrow\infty$. The result follows since ${\displaystyle \lim_{n\rightarrow\infty}p_{n}=\lim_{n\rightarrow\infty}\gamma_{n}=0.}$
\end{proof}


\subsection{ On the containment assumption}
In general, the containment condition is more challenging to check than the
diminishing adaptation condition, and this technical assumption might not even be necessary sometimes 
\citep{roberts2007coupling}. Since the containment assumption is an assumption on the mixing time process
$\{M_\epsilon(\iota_n,\theta_n,X_n),\;n\geq 0\}$, adaptive MCMC algorithms for which the Markov kernels 
$P_{\iota,\theta}$ have similar sufficiently fast mixing behavior typically satisfy the containment assumption.
This intuition was rigorously established by \citep{roberts2007coupling} which showed that the containment
 assumption holds when the family  $\{P_{\iota,\theta}:(\iota,\theta)\in{\bar{\Theta}}\}$ possesses 
a certain simultaneous ergodicity property (either uniform, geometric, or sub-geometric).  
This means in a nutshell that all the Markov kernels $P_{\iota,\theta}$ have qualitatively the same convergence rate, be it uniform, geometric or subgeometric.  Hence to check the containment assumption, we typically need to study the convergence rate of each kernel $P_{\iota,\theta}$,
and verify that all the rates are within a multiplicative constant of each other. Despite recent advances in Markov Chain Monte Carlo theory, there is no general theory that can be easily employed to establish the convergence rate of a given MCMC algorithm. The rate depends as much on the algorithm, as on the target distribution. In this context, it is clear that verifying the containment assumption is a 
mathematically tedious endeavor that is beyond the scope of this paper. We refer the interested 
reader to \citep{roberts2007coupling,andrieu:thoms:08,atchade:fort:10,atchade:etal:11} for some detailed examples.

\section{Example algorithms}\label{method}

We present one specific approach as an example of a Nested Adaptation algorithm. 
Our approach to ``outer adaptation'' will be to identify the ``worst-mixing dimension'' 
(i.e., some parameter or latent state of the statistical model) and update the kernel by assigning different sampler(s) for that dimension. To explain the method, we will give some terminology for describing our algorithm. In particular, we will define a valid kernel, 
MCMC efficiency, and worst-mixing dimension. We will define a set of candidate samplers for a given dimension, which could include scalar samplers or block samplers. 
In either case, a sampler may also have internal adaptation for each parameter or combination. 
To implement the internal clock of each sampler  ($c_{n}$ of the general algorithm), we need to formulate all internal adaptation in the framework using equation \ref{eq:2}. 
We use $P$ (without subscripts) in this section to represent $P_{\iota, \theta}$ 
of the general theory, so the kernel and parameters are implicit.

\subsection{Valid kernel}\label{sec:valid}

Assume our model of interest is $\mathcal{M}$, which could be represented
as a graphical model where vertices or nodes represent states or data
while edges represent dependencies among them. Here we are using ``state''
as Bayesian statisticians do to mean any dimension of the model to be sampled by
MCMC, including model parameters and latent states. We denote the
set of all dimensions of the target distribution, as $\mathcal{X}=\{\mathcal{X}_{1},\ldots,\mathcal{X}_{m}\}$.
Since we will construct a new MCMC kernel as an ordered set of samplers
at each outer iteration, it is useful to define requirements for a
kernel to be valid. We require that each kernel, if used on its own,
would be a valid MCMC to sample from the target distribution $\pi(X),\, X\in\mathcal{X}$
(typically defined from Bayes Rule as the conditional distribution
of states given the data). This is the case if it satisfies the detailed
balance equation, $\pi=P\pi$.

In more detail, we need to ensure that a new MCMC kernel does not
omit some subspace of $\mathcal{X}$ from mixing. Denote the kernel
$P$ as a sequence of (new choice of $c_{n+1}$) samplers $P_{i},\, i=1\ldots j,$
such that $P=P_{j}P_{j-1}\ldots P_{1}$. By some abuse of terminology,
$P$ is a valid kernel if each sampler $P_{i}$ operates on a non-empty
subset $b_{i}$ of $\mathcal{X}$, satisfying ${\displaystyle \bigcup_{i=1}^{j}b_{i}=\mathcal{X}}$.

At iteration $n$, assume the kernel is $P^{(n)}$ and the samples
are $X_{n}=(X_{n,1},\ \ldots,\ X_{n,m})$ where the set of initial
values is $X_{0}$. For each dimension $\mathcal{X}_{k}$, $k=1,\ldots,m$
let $\mathbf{X}_{k}=\{X_{0,k},\ X_{1,k},\ \ldots\}$ be the scalar
chain of samples of $\mathcal{X}_{k}$.

\subsection{Worst mixing state and MCMC efficiency}\label{sec:worst}

We define MCMC efficiency for state $\mathcal{X}_{k}$ from a sample
of size $N$ from kernel $P$ as effective sample size per computation
time 
\[
\omega_{k}(N,P)=\frac{N/\tau_{k}(P)}{t(N,P)},
\]
where $t(N,P)$ is the computation time for kernel $P$ to run $N$
iterations (often $t(N,P)\approx Nt(1,P)$) and $\tau_{k}(P)$ is
the integrated autocorrelation time for chain $\mathbf{X}_{k}$ defined
as 
\[
\tau_{k}=1+2\sum_{i=1}^{\infty}\mathrm{cor}(X_{0,k},X_{i,k}),
\]
\citet{straatsma1986estimation}. The ratio $N/\tau_{k}$ is the effective
sample size (ESS) for state $\mathcal{X}_{k}$ \citep{roberts2001optimal}.
Note that $t(N, P)$ is computation time for the entire kernel, not
just samplers that update $\mathcal{X}_{k}$. $\tau_{k}$ can be interpreted
as the number of effective samples per actual sample. The worst-mixing
state is defined as the state with minimum MCMC efficiency among all
states. Let $k_{min}$ be the index of the worst-mixing state, that
is 
\[
k_{min}=\arg\min_{k}\tau_{k}^{-1}.
\]
Since the worst mixing dimension will limit the validity of the entire
posterior sample \citep{thompson2010graphical}, we define the efficiency
of a MCMC algorithm as $\omega_{k_{min}}(N,P)$, the efficiency of
the worst-mixing state of model $\mathcal{M}$.

There are several ways to estimate ESS, but we use $\texttt{effectiveSize}$
function in the R coda package \citep{plummer2006coda} since this
function provides a stable estimation of ESS. This method, 
which is based on the spectral density at frequency zero, 
has been proven to have the highest convergence rate, thus giving a more accurate
and stable result \citep{thompson2010graphical}.

\subsection{Candidate Samplers}\label{sec:candidate}

A set of candidate samplers $\{P_{j},\: j\in\mathcal{S}\}$ is a list
of all possible samplers that could be used for a parameter of the model
$\mathcal{M}$. These may differ depending on the parameter's characteristics and role in the model (e.g., whether there is a valid Gibbs sampler,
or whether it is restricted to $[0,\infty)$). In addition to univariate
candidate samplers, nodes can also be sampled by block samplers. Denote
$|b|$ the number of elements of $b.$ If $|b_{i}|>1$, $P^{(n)}_{i}$,
the sampler applied to block $b_{i}$ at iteration $n$, is called
a block sampler; otherwise it is a univariate or scalar sampler.

In the examples below we considered up to four univariate candidate
samplers and three kinds of block samplers. The univariate samplers
included adaptive Metropolis-Hastings (AMH), adaptive Metropolis-Hastings on a log
scale (AMHLS) for states taking only positive real values, Gibbs samplers
for states with a conjugate prior-posterior pairing, and slice samplers.
The block samplers included adaptive Metropolis-Hastings with multivariate
normal proposals, automated factor slice sampler \citep{tibbits2014automated} (slice samplers in
a set of orthogonal rotated coordinates), and automated factor random
walk (univariate random walks in a set of orthogonal rotated coordinates).
These choices are by no means exhaustive but serve to illustrate the
algorithms here.
\subsubsection{Block samplers and how to block}\label{sec:block}

\citet{turek2017automated} suggested different ways to block the states
efficiently: (a) based on correlation clustering, (b) based on model
structure. Here we use the first method. 

At each iteration, we use the generated samples to create the empirical
posterior correlation matrix. To stabilize the estimation, all of the samples
are used to compute a correlation
matrix $\rho_{d\times d}$. This in turn is used to make a distance
matrix $D_{d\times d}$ where $D_{i,j}=1-|\rho_{i,j}|$ for $i\neq j$
and $D_{i,i}=0$ for every $i$, $j$ in $1,\ldots,d$. To ensure
minimum absolute pairwise correlation between clusters, we construct a
hierarchical cluster tree from the distance matrix $D$ (\citet{everitt2011hierarchical} chapter 4). Given a selected
height, we cluster the hierarchical tree into distinct groups of states.
Different parts of the tree may have different optimal heights for forming blocks. Instead of using a global height to cut the tree, we only choose a block that
contains the worst-mixing state from the cut and keep the other
samplers unchanged. At each outer iteration, if the algorithm is choosing a new block sampler as part of a new kernel, it will cut the tree at a lower correlation threshold, thus creating a cluster of nodes with lower posterior correlations. In our implementation, we use the R function $\texttt{hclust}$
to build the hierarchical clustering tree with ``complete linkage''
from the distance matrix
$D$. By construction, the absolute correlation between states within
each group is at least $1-h$ for $h$ in $[0,1]$. We then use the R
function $\texttt{cutree}$ to choose a block that contains the worst-mixing state. 
This process is justified in the sense that the partitioning
adapts according to the model structure through the posterior correlation.
The details and validity of the block sampling in our general framework are provided in Appendix A. 

\subsection{How to choose new samplers}\label{sec:choice}
To choose new samplers to compose a new kernel, we determine the worst-mixing 
state and choose randomly
from candidate samplers to replace whatever sampler was updating it
in the previous kernel while keeping other samplers the same. There
are some choices to make when considering a block sampler. If the worst-mixing parameter is $x$, 
and the new kernel will use a block sampler for $x$ together with one or more parameters $y$, 
we keep the current sampler(s) used for $y$. 
Future work can consider other schemes such as changing group of samplers together
based on model structure.

\subsection{Internal clock variables}
\label{sec:internal} In the algorithm \ref{alg1},
$\theta$ represents the internal adaptation parameter of a particular
sampler and $c$ represents its internal clock. In general, an internal
clock variable is defined as a variable used in a sampler to determine
the size of internal adaptation steps such that any internal adaptation
would converge in a typical MCMC setting. An example of an internal clock
variable is a number of internal iterations that have occurred. To
use a sampler in the general framework, we need to establish what
are its internal adaptation and clock variables. A few examples of
internal adaptation variables of different samplers are summarized
as follows:
\begin{itemize}
\item For adaptive Metropolis-Hastings: proposal scale is used.
\item For block adaptive Metropolis-Hastings: proposal scale and covariance matrix
are used.
\item For automated factor slice sampler: covariance matrix (or equivalent,
i.e. coordinate rotation) is used.
\item For automated factor random walk: covariance matrix (ditto) and proposal
scales for each rotated coordinate axis are used. 
\end{itemize}

These internal adaption variables are set to default initial values when its sampler is first used. 
After that, they are retained along with internal clock variables so that whenever 
we revisit a sampler, 
we will use the stored values to set up this sampler.
This setting guarantees the diminishing adaption property, 
which is essential for the convergence of the algorithm. 
Pseudo-code for Nested Adaptation MCMC is given in Algorithm 1.

\begin{algorithm}[htp]
\caption{Nested Adaptation MCMC} \label{alg1}
\small
\begin{algorithmic}[1]
\INPUT  
\Statex Bayesian model with initial state (including latent variables) ${X}_{0}$ 
\Statex $\{p_{n}, n \in \mathbb{N} |p_n\in(0,1),\lim_{n}p_{n}=0\}$, maximum iteration $M$
\Statex Candidate samplers $\{P_{j},  j \in \mathcal{S}\}$ 
\Statex $P_{\iota_0, \theta_0}:=$ ordered set of initial samplers $\{P^{(0)}_j\}_{j\in \mathcal{S}}$ from Bayesian model  
\Ensure
\Statex An ordered set of samplers $\{P_{i^*}\}_{i^*\in \mathcal{S}}$ with the best MCMC efficiency so far 
\stb Initialize $\mathrm{EFF}$, $\mathrm{EFF_{best}}$, $n$, $\kappa_0$, $\tau_0$, $c_0$  to $0$ \Comment Denote MCMC efficiency $\mathrm{EFF}$
\While {($\mathrm{EFF}$ $\ge$ $\mathrm{EFF_{best}}$) or ($n < M$)} 
  \State Sample $N$ samples from the current sampler set $\{P_j^{(n)}\}_{j\in \mathcal{S}}$
  \State Store internal clocks $c_n$ and adaption variables $\theta_{n}$ for each sampler \Comment Section \ref{sec:internal}
  \State Compute $\mathrm{EFF_k}=\omega_{k}(N,P)=\frac{N/\tau_{k}(P)}{t(N,P)}$	\Comment $k$ is an index of parameters
  \State Identify $k_{min}=\arg\min_{k}\tau_{k}^{-1}$, $\mathrm{EFF}=\mathrm{EFF_{k_{min}}}$  \Comment See Section \ref{sec:worst}
  \If {($\mathrm{EFF}$ $\ge$ $\mathrm{EFF_{best}}$)} 
      \State Set $\{P_{i^*}\}_{i^*\in \mathcal{S}}=\{P_i^{(n)}\}_{i\in \mathcal{S}}$
      \State Set $\mathrm{EFF_{best}}=\mathrm{EFF}$
  \Else
      \State Set $\{P_i^{(n)}\}_{i\in \mathcal{S}}=\{P_{i^*}\}_{i^*\in \mathcal{S}}$ 
  \EndIf
  \State Draw $B_{n+1} \sim \mathrm{Bern}(p_{n+1})\in \{0,1\}$
  \If {$B_{n+1}=0$} 
    \State $\kappa_{n+1}= \kappa_{n}$, $\tau_{n+1}=\tau_n+1$, $c_{n+1}=\kappa_{n+1}+\tau_{n+1}$
  \Else
    \State $\kappa_{n+1}= \kappa_{n}$+1, $\tau_{n+1}=0$, $c_{n+1}=\kappa_{n+1}+\tau_{n+1}$ 
    \State Set $P^{(n+1)}_{i}=P^{(n)}_{i}, i \neq k_{\mathrm{min}}$, choose $P^{(n+1)}_{k_{\mathrm{min}}}$ from candidate samplers \Comment Section \ref{sec:choice}   
    \If {($P^{(n+1)}_{k_{\mathrm{min}}}$ has been used before)} 
      \State Use $c_n$, $\theta_{n}$ to set up the sampler $P^{(n+1)}_{k_{\mathrm{min}}}$ 
    \Else
      \State Use default internal adaptation value of $P^{(n+1)}_{k_{\mathrm{min}}}$ \Comment Section \ref{sec:internal}
    \EndIf
  \EndIf
  \State Set $n = n+1$
\EndWhile
\end{algorithmic}
\end{algorithm}

\section{Examples}\label{example}

In this section, we evaluate our algorithm on some benchmark examples
and compare them to different MCMC algorithms. In particular, we compare
our approach to the following MCMC algorithms.
\begin{itemize}
\item All Scalar algorithm: Every dimension is sampled using an adaptive
scalar normal random walk sampler. 
\item All Blocked algorithm: All dimensions are sampled in one adaptive
multivariate normal random walk sampler. 
\item Default algorithm: Samplers are assigned as follows.  When possible, 
a conjugate (Gibbs) sampler is used. Otherwise, adaptive random-walk Metropolis-Hastings samplers are used.
These will be block samplers for parameters arising from multivariate distributions and scalar samplers otherwise.  
\item Auto Block algorithm: The Auto Block method
  \citep{turek2017automated} searches blocking schemes based on
  hierarchical clustering from posterior correlations to determine a
  highly efficient (but not necessarily optimal) set of blocks that
  are sampled with multivariate normal random-walk samplers. Thus,
  Auto Block uses only either scalar or multivariate adaptive
  random walk, concentrating more on partitioning the correlation
  matrix than trying different sampling methods.  Note that the
  initial sampler of both the Auto Block algorithm and our
  proposed algorithm is the All Scalar algorithm.
\end{itemize}
All experiments were carried out using the NIMBLE package \citep{nimble2017}
for R \citep{R2013} on a cluster using $32$ cores of Intel Xeon E5-2680
$2.7$ Ghz with $256$ GB memory. Models are coded using NIMBLE's
version of the BUGS model
declaration language \citep{lunn2000winbugs,lunn2012bugs}. All MCMC
algorithms are written in NIMBLE, which provides user-friendly interfaces
in R and efficient execution in custom-generated C++, including matrix operations in the C++ Eigen
library \citep{guennebaud2010eigen}.

To measure the performance of an MCMC algorithm, we use MCMC
efficiency.  MCMC efficiency depends on ESS, estimates of which can have a high variance for a short Markov chain.  This presents
a tuning-parameter trade-off for the Nested Adaptation method: Is it better to
move cautiously (in sampler space) by running long chains for each outer adaptation in
order to gain an accurate measure of efficiency, or is it better to
move adventurously by running short chains, knowing that some
algorithm decisions about samplers will be based on noisy efficiency
comparisons?  In the latter case, the final samplers may be less optimal,
but that may be compensated by the saved computation time.  To explore this
trade-off,  we try our Nested Adaptation algorithm with different sample
sizes in each outer adaptation and label results accordingly.  For
example, Nested Adaptation 10K will refer to the Nested Adaptation method with
samples of 10,000 per outer iteration. 

We present algorithm comparisons in terms of time spent in an
adaptation phase, final MCMC efficiency achieved, and the time
required to obtain a fixed effective sample size (e.g., 10,000).  Only
Auto Block and Nested Adaptation have adaptation phases.  An important
difference is that Auto Block did not come with a proof of valid
adaptive MCMC convergence (it could be modified to work in the current
framework, but we compare to the published version).  Therefore,
samples from its adaptation phase are not normally included in the
final samples, while the adaptation samples of Nested Adaptation can be
included.  

To measure final MCMC efficiency, we conducted a single long run of
length $N$ with the final kernel of each method solely for the purpose
of obtaining an accurate ESS estimate.  One would not normally do such
a run in a real application.  The calculation of time to obtain a
fixed effective sample size incorporates both adaptation time and
efficiency of the final samplers. For both Nested Adaptation and Auto Block,
we placed them on a similar playing field by assuming for this
calculation that samples are not retained from the adaptation phase,
making the results conservative.

For all comparisons, we used $20$ independent runs of each method and
present the average results from these runs.  To show the variation in
runs, we present box-plots of efficiency in relation to computation
time from the $20$ runs of Nested Adaptation.  The final (right-most) box-plot
in each such figures shows the $20$ final efficiency estimates from
larger runs.  Not surprisingly, these can be lower than obtained by
shorter runs.  These final estimates are reported in the tables.

A public Github repository containing scripts for reproducing our
results may be found at \href{https://github.com/nxdao2000/AutoAdaptMCMC}{https://github.com/nxdao2000/AutoAdaptMCMC}. 
Some additional experiments are also provided there.

\subsection{Linear state space model}

First, to illustrate the distinction between the Nested Adaptation algorithm
and other methods, we will use a linear state space model
\citep{durbin2012time} in which all parameters are fixed. Such
a model will allow us to provide a simple assessment
of the performances of the methods for a range of situations, including
those in which the models were true.

Let $x_{t}$ be the latent state at time $t$, $y_{t}$ be the observed data, and suppose we have the number of time
points $T=50$. Let the initial priors be: 
\[
a\sim Uniform(0,1),
\]
\[
b\sim N(1,1),
\]
\[
\mu_{0}\sim N(0,1),
\]
and observation density $y_{t}\sim N(bx_{t},\ 1)$ for $t=2,\ \ldots,50.$
State transitions are governed by a first order autoregressive (AR) process
\[
x_{t}\sim N(ax_{t-1},1)\quad for\;t=2,...,50.
\]

\begin{table}
\begin{center}
\begin{tabular}{lrrrrr}
 \hline
  Algorithms& \rule[-1.5mm]{0mm}{0mm} Adapt time
 & Efficiency & Time to $10^4$ effective samples\\ 
  \hline
All Blocked  &  0.00  & 52.359 & 1910\\ 
Default  & 0.00    & 1198.128 & 83\\ 
All Scalar  & 0.00    & 1218.060 & 82\\
Auto Block  & 22.92  & 1277.639 & 101 \\ 
Nested Adaptation 5K  & 0.9    & 1796.079  & 56 \\
Nested Adaptation 10K  & 1.7   & 1693.929 & 61  \\  
Nested Adaptation 20K  & 3.3   & 1616.752 & 65  \\  
\hline
\end{tabular}
\caption{Summary results of different MCMC algorithms for the linear state space model. Runtime is presented
as seconds, and efficiency is in units of effective samples produced
per second of algorithm runtime. Time to $N$ effective samples is computed by $N/\mathrm{efficiency}$ for 
static algorithms and that plus adaptation time for Auto Block and Nested Adaptation algorithms.}
\label{table:1}
\end{center}
\end{table}

\begin{figure}
\begin{centering}
\vspace{-0.5cm}
 \includegraphics[width=6.7cm,height=7cm]{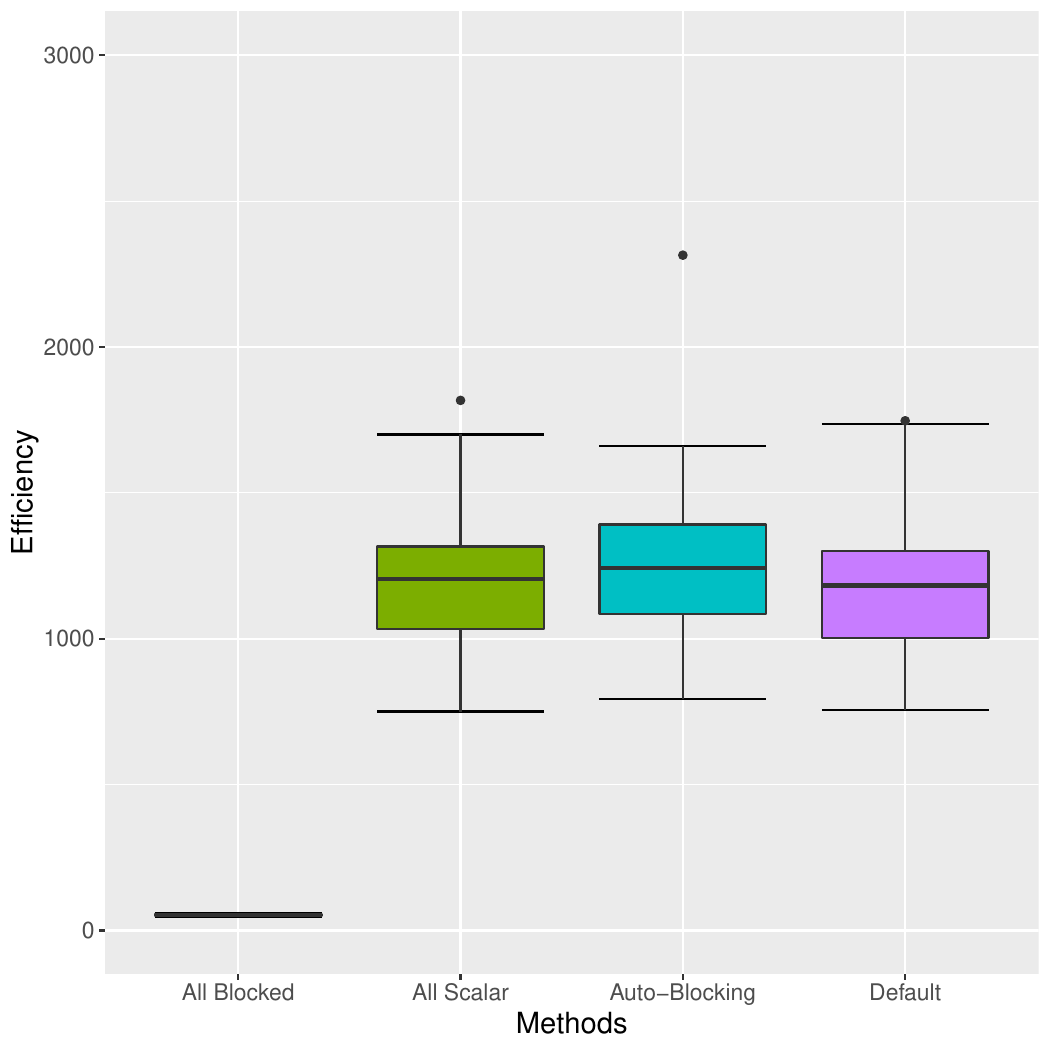}\includegraphics[width=6.5cm,height=7cm]{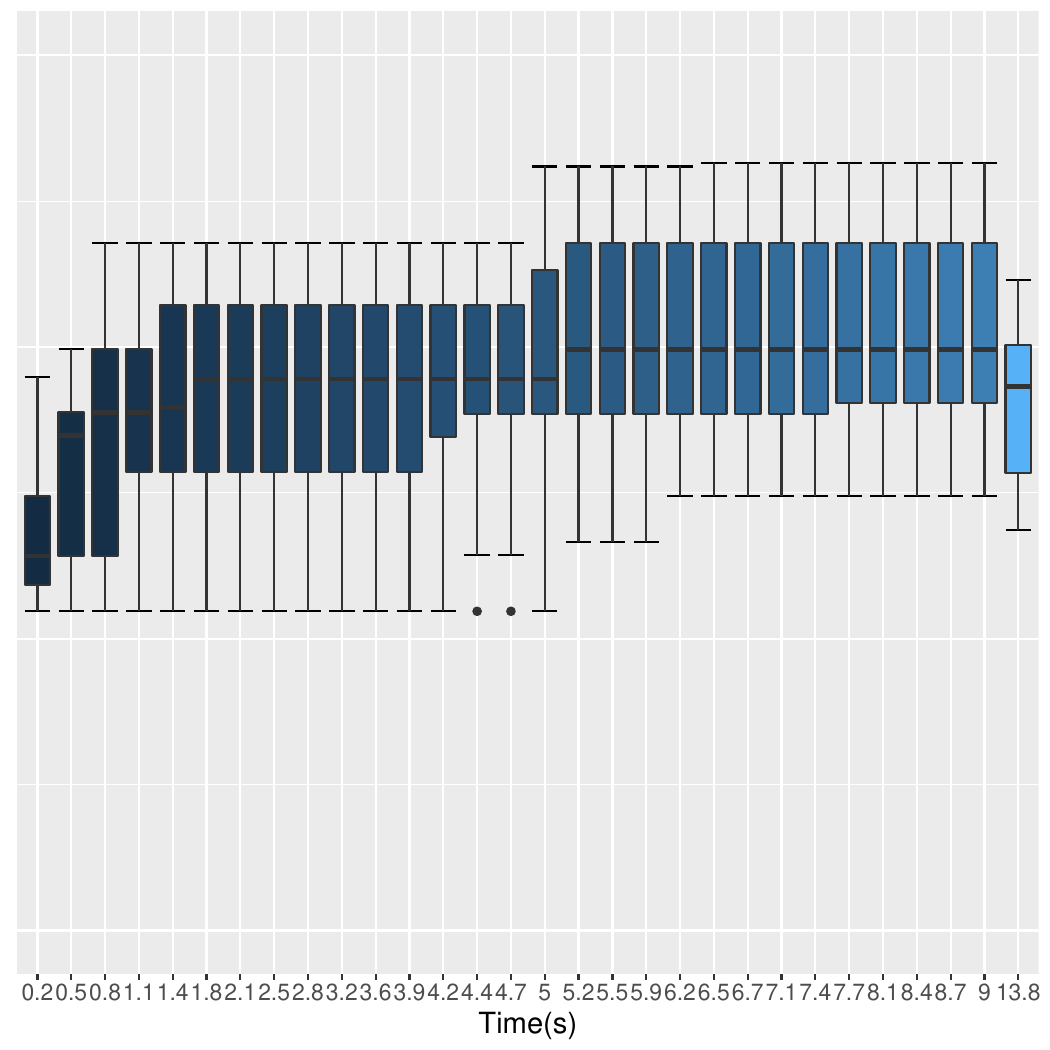}
\par\end{centering}
\caption{MCMC efficiencies of different methods for linear state space model. The left panel shows the box-plots of
MCMC efficiencies of All Blocked, All Scalar, Auto Block and Default algorithms computed from 20 replications.
The right panel shows the box-plots of MCMC efficiencies of Nested Adaptation 5K algorithm computed from 20 replications at each outer adaptation. The last (right-most) box-plot is
computed from the chain of 100000 samples generated. The time axis shows the average computational time of 20 replications.}
\label{fig:linear} 
\end{figure}

$x_{1}$ and $y_{1}$ were randomly generated from simulating $x_{1}\sim N(\mathrm{\mu_{0}},\ 1)$,
$y_{1}\sim N(x_{1},\ 1)$ with $\mu_{0}=0$ while $x_{t}$ and $y_{t}$
are subsequently generated using transition density with the parameter values
$a=0.2$ and $b=2$. These values were arbitrarily chosen to cover
the range observed for the explanatory variables; we desire high efficiency regardless of model fit, 
so the particular choice of this simple model is tangential to our main points.  

In this relatively simple example, we try our Nested Adaptation
algorithm with large numbers of iterations in each outer adaptation
as well as a large number of iterations for the final efficiency.
Specifically, we used samples sizes of 5000, 10000 and 20000 per
outer adaptation and we used $N=100000$ for computing final efficiency.

For this example (Table \ref{table:1}) All Scalar sampling produces
MCMC efficiency of about $121.8$, while the All Blocked algorithm,
which consists of a single block sampler of dimension $53$, has
MCMC efficiency of approximately $5.2$. All Blocked
samples all $53$ dimensions jointly, which requires computation
time roughly double that of All Scalar and yields only rather low
ESS. The Default algorithm performs similarly
to Auto Block but worse than Nested Adaptation algorithm. Since Auto Block converges to All Scalar in this case, the Auto Block algorithm performs no better than standard All Scalar (Figure~\ref{fig:linear}) as would be expected.
It is clear that Nested Adaptation method has dramatic improvements
even when we take into account the adaptation time. Amongst these
Auto methods, Auto Block performs worse than all Nested Adaptation 5K, Nested Adaptation 10K and Nested Adaptation 20K
in both computational time and MCMC efficiencies. Overall, Nested Adaptation
5K appears to be the most efficient method in terms of time to 10000
effective samples. One interpretation is that Nested Adaptation 5K trades
off well between adaptation time and MCMC efficiency in this model.
The final samplers from Nested Adaptation included an automated factor slice sampler on block (a, b), 
an automated factor random walk sampler on block ($\mu_0$, $x_1$) and adaptive random-walk Metropolis-Hastings samplers on other nodes.

Figure \ref{fig:linear} illustrates that Nested Adaptation algorithm outperforms the other
algorithms. This comes from both the flexibility
to trade-off the number of outer adaptations vs. adaptive time to
reach a good sampler as well as the larger space of kernels being
explored. Since MCMC efficiency is highly dependent upon hierarchical
model structures, using scalar and multivariate normal random walks
alone, as done by the Auto Block algorithm, can be quite limiting.
Nested Adaptation, can overcome this limitation with the flexibility to choose
different type of samplers. We will see that more strongly in the
next examples, where the models are more complex.

\subsection{A random effect model}

We consider the ``litters'' model, which is an original example
model provided with the MCMC package WinBUGS. This model is chosen
because of its notoriously slow mixing, which is due to the strong
correlation between parameter pairs. It
is desirable to show how much improvement can be achieved compared to other approaches
on this benchmark example. The purpose of using this simple example is to establish the potential utility of the Nested Adaptation approach, while saving more advanced applications for future work. 
In this case, we show that our algorithm
indeed outperforms by a significant margin the other approaches. This
model's specification is given following \citet{deely1981bayes} and
\citet{kass1989approximate} as follows.

\begin{table}
\begin{center}
\begin{tabular}{lrrrrr}
 \hline
 Algorithms& \rule[-1.5mm]{0mm}{0mm} Adapt time
 & Efficiency & Time to $10^4$ effective samples\\ 
  \hline
All Blocked  &  0.00  & 0.5855 & 17079\\ 
Default  & 0.00    & 1.8385 & 5439\\ 
All Scalar  & 0.00    & 1.6870 & 5928\\
Auto Block  & 21.97  & 12.1205 & 847 \\ 
Nested Adaptation 50K  & 4.09    & 16.4532 & 612 \\
Nested Adaptation 70K  & 5.95    & 18.0358 & 560 \\
Nested Adaptation 100K  & 8.34   & 17.6278 & 576  \\  
\hline
\end{tabular}
\caption{Summary results of different MCMC algorithms for the litters model. Runtime is presented
as seconds, and efficiency is in units of effective samples produced
per second of algorithm runtime. Time to $N$ effective samples is computed by $N/\mathrm{efficiency}$ for 
static algorithms and that plus adaptation time for Auto Block and Nested Adaptation algorithms.}
\label{table:2}
\end{center}
\end{table}

\begin{figure}
\begin{centering}
\vspace{-0.5cm}
 \includegraphics[width=6.5cm,height=7cm]{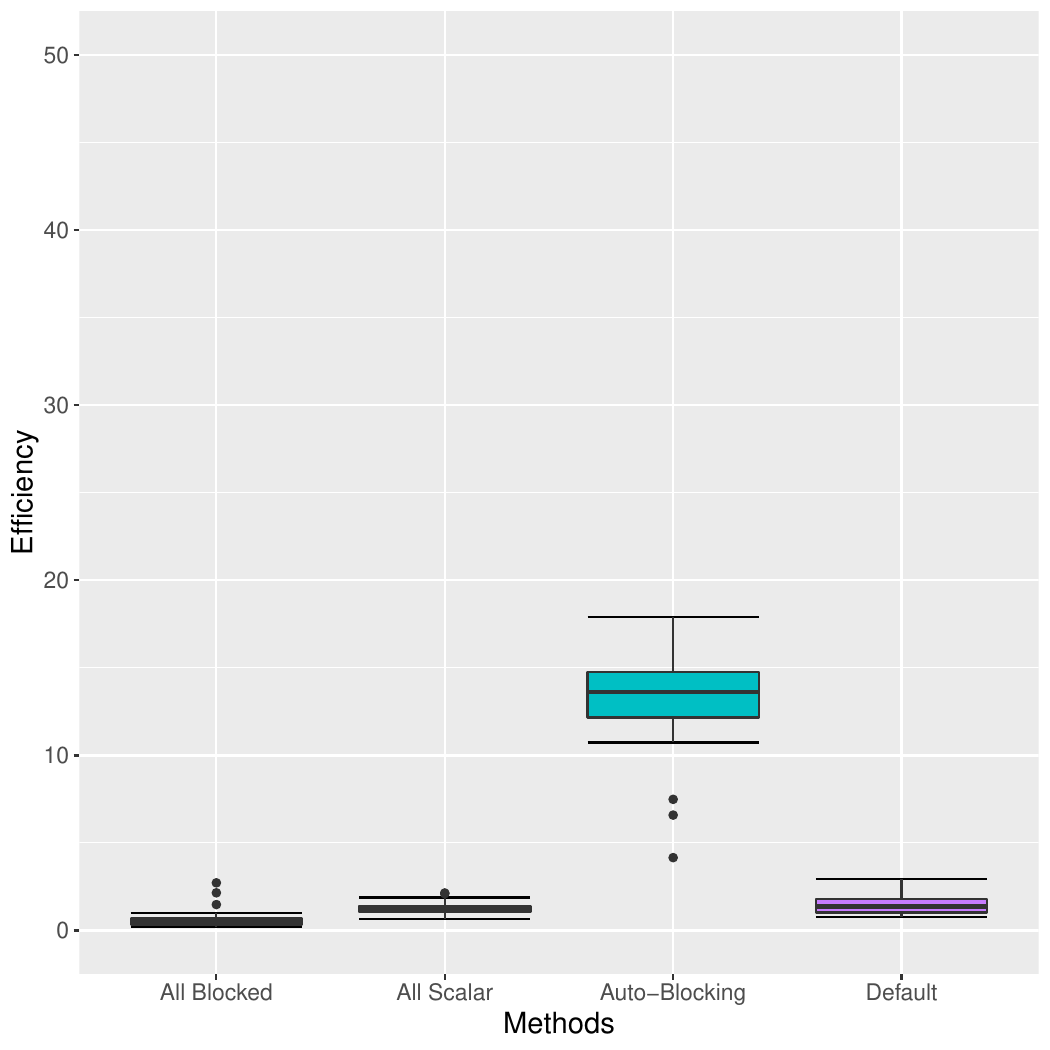}\includegraphics[width=6.5cm,height=7cm]{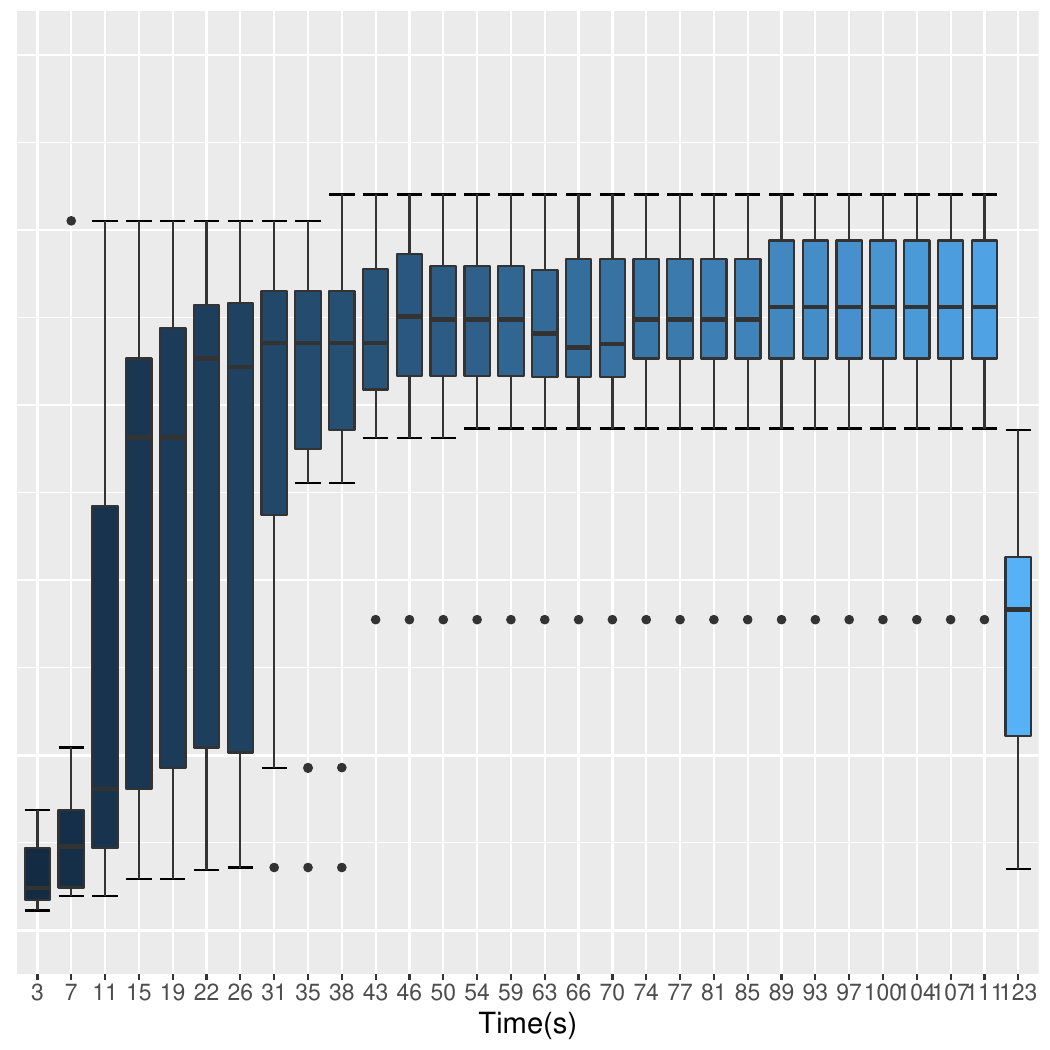}
\par\end{centering}
\caption{MCMC efficiencies of different methods for litters model. The left panel shows the box-plots of
MCMC efficiencies of All Blocked, All Scalar, Auto Block and Default algorithms computed from 20 replications.
The right panel shows the box-plots of MCMC efficiencies of Nested Adaptation 50K algorithm computed from 20 replications at each outer adaptation. The last (right-most) box-plot is
computed from the chain of 300000 samples generated. The time axis shows the average computational time of 20 replications.}
\label{fig:toy1} 
\end{figure}

Suppose we observe the data in $i$ groups. In each group, the data
$y_{ij}$, $j=\{1,\ldots,n\}$ are conditionally independent given
the parameters $p_{ij}$, with the observation density 
\[
y_{ij}\sim \mathrm{Bin}(n_{ij},p_{ij}).
\]
In addition, assume that $p_{ij}$ for fixed $i$ are conditionally
independent given the ``hyper-parameters'' $\alpha_{i}$, $\beta_{i}$, with conjugate density 
\[
p_{ij}\sim \mathrm{Beta}(\alpha_{i},\beta_{i}).
\]
Assume that $\alpha_{j}$, $\beta_{j}$ follow prior densities, 
\[
\alpha_{1}\sim \mathrm{Gamma}(1,0.001),
\]
\[
\beta_{1}\sim \mathrm{Gamma}(1,0.001),
\]
\[
\alpha_{2}\sim \mathrm{Uniform}(0,100),
\]
\[
\beta_{2}\sim \mathrm{Uniform}(0,50).
\]

Following the setup of \citet{rue2005gaussian} as discussed in \citet{turek2017automated},
we could jointly sample the top-level parameters and conjugate latent
states as the beta-binomial conjugacy relationships allow the use 
of what \citet{turek2017automated} call cross-level sampling, but, for demonstration purposes, we do not
include this here. 

Since the litters model mixes poorly, we run a large
number of iterations (i.e. $N=300000$) to produce stable estimates of
final MCMC efficiency. We start both Auto Block and Auto
Adapt algorithms with All Scalar and adaptively explore the
space of all given candidate samplers.   We use Nested Adaptation with
either 50000, 70000 or 100000 iterations per outer adaptation.

Results (Table \ref{table:2}) show that Auto Block generates
samples with MCMC efficiency about seven-fold, six-fold and
twenty-fold that of the All Scalar, Default and All
Blocked methods, respectively. We can also see that as the outer
adaptation sample size increases, the performance of Nested Adaptation
improves.  Final MCMC efficiencies of Nested Adaptation 50K, Nested Adaptation
70K and Nested Adaptation 100K are 135\%, 148\% and 145\% of MCMC
efficiency of Auto Block, respectively. In addition, the
adaptation time for all cases of Nested Adaptation are much shorter than
for Auto Block.  Combining adaptation time and final efficiency
into the resulting time to 10000 effective samples, we see that in this
case, larger samples in each outer iteration are worth their
computational cost.
In this example, the final samplers from Nested Adaptation included automated factor random walk samplers on blocks 
($\alpha_1$, $\beta_1$) and ($\alpha_2$, $\beta_2$), a slice sampler on $p_{1,16}$ and adaptive Metropolis-Hastings samplers on all other nodes

Figure \ref{fig:toy1} shows the box-plots
computed from 20 independent runs on the litters model of All Blocked,
All Scalar, Auto Block, Default and Nested Adaptation 50K
algorithms. The left panel of the figure confirms that MCMC efficiency
of Auto Block is well dominated that of other static adaptive
algorithms. The right panel of the figure shows the MCMC efficiency of
Nested Adaptation 50K gradually improves with time. The right-most box-plot
verifies that the MCMC efficiency of selected samplers from
Nested Adaptation algorithm (computed from large samples) is better than that of Auto Block algorithm. Last but not least,
Nested Adaptation algorithms are much more efficient than Auto Block in
the sense that we can keep every sample while Auto Block algorithm
throws away most of the samples.

\subsection{Spatial model}

In this section, we consider a hierarchical spatial model as the final
example. We use the classical scallops dataset for this model. This
dataset is chosen since we want to compare our approach with other
standard approaches in the presence of spatial dependence.  This data
collects observations of scallop abundance at 148 locations from the New
York to New Jersey coastline in 1993. It was surveyed by the Northeast
Fisheries Science Center of the National Marine Fisheries Service and
made publicly available at
http://www.biostat.umn.edu/\textasciitilde{}brad/data/myscallops.txt.
It has been analyzed many times, such as
\citet{ecker1994geostatistical,ecker1997bayesian,banerjee2014hierarchical}
and references therein. Following
\citeauthor{banerjee2014hierarchical}, assume the log-abundance
$\mathbf{g}=(g_{1},\ldots,g_{N})$ follows a multivariate normal
distribution with mean $\bm{\mu}$ and covariance matrix $\mathbf{\Sigma}$,
defined by covariances that decay exponentially as a function of distance. Specifically,
let $y_{i}$ be measured scallop abundance at site $i$,
$d_{i,j}$ be the distance between sites $i$ and $j$, and $\rho$ be a
valid correlation. Then
\[
\mathbf{g}\sim\mathrm{N}\left(\bm{\mu},\mathbf\Sigma\right),
\]
where each component $\Sigma_{ij}=\sigma^{2}exp(-d_{i,j}/\rho).$

We model observations as $y_{i} \sim
\mathrm{Poisson}(\mathrm{exp}(g_{i}))$.  Priors for $\sigma$ and
$\rho$ are Uniform over a large range of interest.
The parameters in the posterior distribution are expected to be correlated
as the covariance structure induces a trade-off between $\sigma$
and $\rho$. This can be sampled well by Auto Block algorithm,
and we would like to show that our approach can achieve even higher
efficiency with lower computational cost of adaptation.

This spatial model, with 858 parameters, is computationally expensive
to estimate. Therefore, we will use Nested Adaptation 5K, Nested Adaptation
10K and Nested Adaptation 20K algorithms for comparison and run
$N=50000$ for estimating final efficiency. 
In this case, the typical final samplers from Nested Adaptation included an automated factor random walk sampler 
on block ($\sigma$, $\rho$), adaptive random walk samplers on blocks ($\mu$, $g_{34}$, $g_{66}$, $g_{148}$) and ($g_{39}$, $g_{52}$), an 
automated factor slice sampler on block ($g_{11}$, $g_{12}$, $g_{33}$, $g_{48}$, $g_{92}$), 
slice samplers on $g_{54}$, $g_{65}$, and $g_{68}$, and adaptive random-walk Metropolis-Hastings samplers on the rest.

As can be seen from Table \ref{table:3}, All Blocked and
Default algorithms mix very poorly, resulting in extremely low
efficiencies of 0.01 and 0.002, respectively. The All Scalar
algorithm, while achieving higher ESS, runs slowly because large matrix
calculations are needed for every univariate sampler. The Auto Block
algorithm, on the other hand, selects an optimal threshold to cut the
entire hierarchical clustering tree into different groups, increasing
the ESS about 3 times. With a few small blocks, the computation cost of
Auto Block is somewhat cheaper than All Scalar algorithm. As a
result, the efficiency mean is about 3.5 times that of All
Scalar. Meanwhile, our Nested Adaptation 5K, 10K and 20K algorithms perform
best.  It should be noted that the Nested Adaptation algorithm can achieve
good mixing with adaptation times that are only 15.5\%, 32.5\% and
59\% compared to the adaptation time of Auto Block. In Figure 3,
while the left panel shows a distinction between Auto Block and
other static algorithms, the right panel shows that Nested Adaptation 20K
surpasses Auto Block in just a few outer iterations, indicating
substantial improvements in some models.

\begin{table}
\begin{center}
\begin{tabular}{lrrrrr}
 \hline
 Algorithms& \rule[-1.5mm]{0mm}{0mm} Adapt time
 & Efficiency & Time to $10^4$ effective samples\\ 
  \hline
All Blocked  &  0.00  & 0.0100 & 1000000\\ 
Default  & 0.00    & 0.0020 & 5000000\\ 
All Scalar  & 0.00    & 0.1150 & 86956\\
Auto Block  & 19094.89  & 0.3565 & 47145\\ 
Nested Adaptation 5K  & 2967.55    & 0.4420 &  25592 \\
Nested Adaptation 10K  & 6221.61  & 0.4565 & 28127 \\ 
Nested Adaptation 20K  & 11278.78   & 0.4948  &  31488 \\  
\hline
\end{tabular}
\caption{Summary results of different MCMC algorithms for spatial model. Runtime is presented
as seconds, and efficiency is in units of effective samples produced
per second of algorithm runtime. Time to $N$ effective samples is computed by $N/\mathrm{efficiency}$ for 
static algorithms and that plus adaptation time for Auto Block and Nested Adaptation algorithms.}
\label{table:3}
\end{center}
\end{table}

\begin{figure}
\begin{centering}
\vspace{-0.5cm}
 \includegraphics[width=6.5cm,height=7cm]{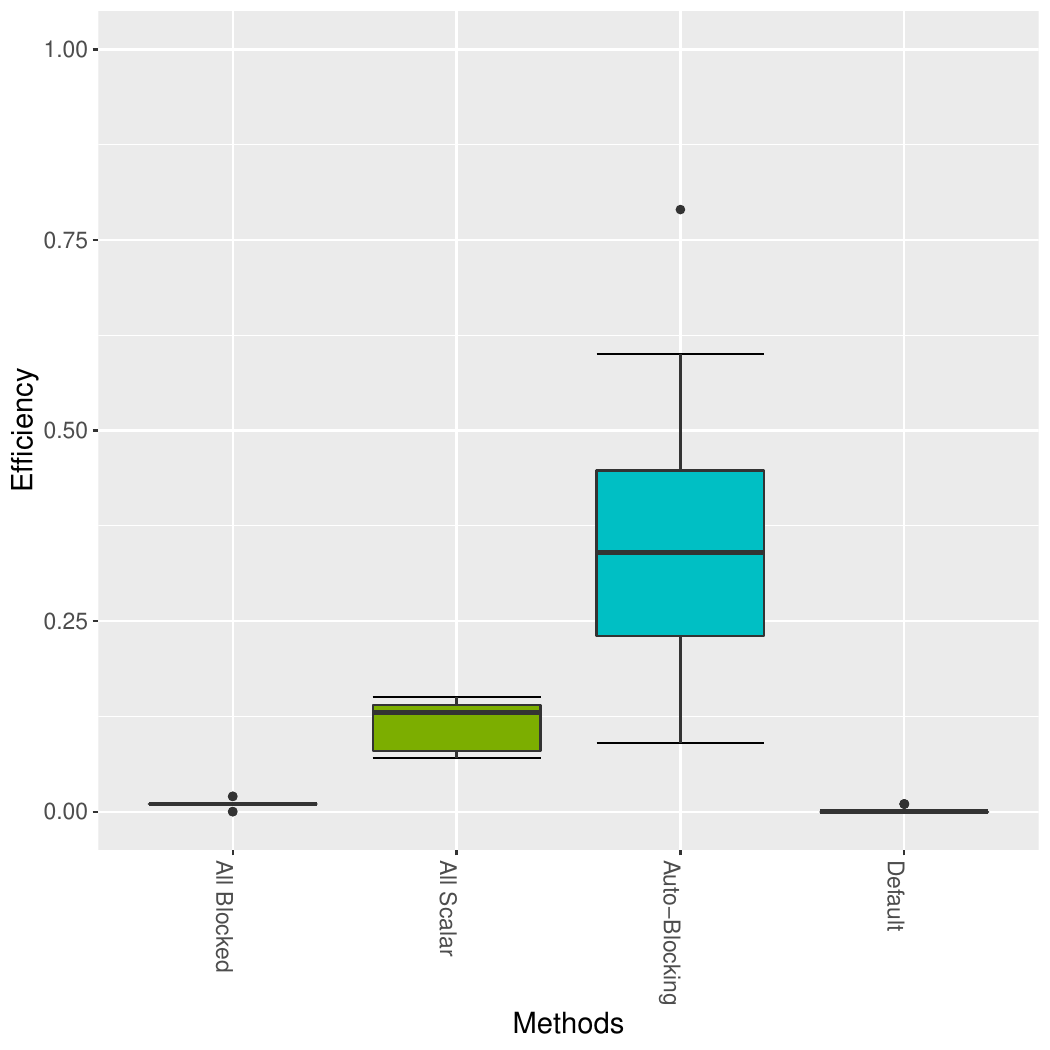}\includegraphics[width=6.5cm,height=7cm]{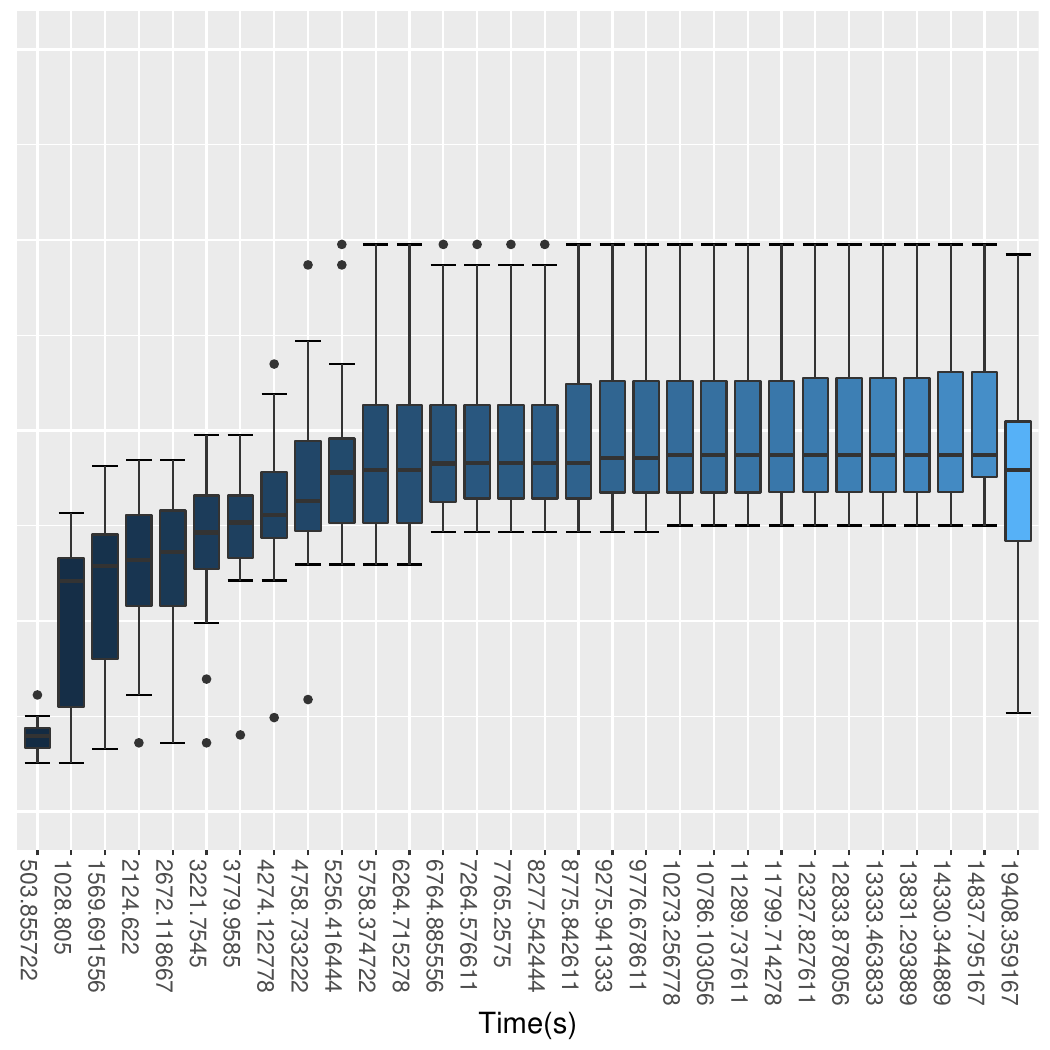}
\par\end{centering}
\caption{MCMC efficiencies of different methods for spatial model. The left panel shows the box-plots of
MCMC efficiencies of All Blocked, All Scalar, Auto Block and Default algorithms computed from 20 replications.
The right panel shows the box-plots of MCMC efficiencies of Nested Adaptation 5K algorithm computed from 20 replications at each outer adaptation. The last (right-most) box-plot is
computed from the chain of 10000 samples generated. The time axis shows the average computational time of 20 replications.}
\label{fig:spatial} 
\end{figure}

\section{Discussion}
\label{discussion} We have proposed a general Nested Adaptation MCMC
algorithm. Our algorithm traverses a space of valid MCMC kernels to
find an efficient algorithm automatically. There is only one previous
approach, namely Auto Block sampling, of this kind that we are
aware of. We have shown that our approach can substantially outperform
Auto Block in some cases, and that both outperform simple static
approaches.  Using some benchmark models, we observe that our
approach can yield improvements, which can be substantial compared to the Default method.

The comparisons presented have deliberately used fairly simple
samplers as options for Nested Adaptation in order to avoid comparisons among
vastly different computational implementations.  A major feature of
our framework is that it can incorporate almost any sampler as a
candidate and almost any strategy for choosing new kernels from
compositions of samplers based on results so far.  Samplers to be
explored in the future could include auxiliary variable algorithms such
as slice sampling or derivative-based sampling algorithms such as
Hamiltonian Monte Carlo \citep{duane1987hybrid}.  Now that the basic
framework is established and shown to be useful in simple cases, it
merits extension to more advanced cases.

The Nested Adaptation method can be viewed as a generalization of the
Auto Block method.  It is more general in the sense that it can use
more kinds of samplers and explore the space of samplers more
generally than the correlation-clustering of Auto Block.  Thus, our framework
can be considered to provide a broad class of automated kernel construction
algorithms that use a wide range of sampling algorithm as components.

If block sampling is included in the space of the candidate samplers,
choosing optimal blocks is important and can greatly increase the
efficiency of the algorithm. For this reason, we extended the cutting
of a hierarchical cluster tree to allow different cut heights on
different branches (different parts of the model). This differs from
Auto Block, which forms all blocks by cutting the entire tree at the
same height. We also have different multivariate adaptive
sampling other than random walk normal distribution
such as automated factor slice sampler and automated factor random walk sampler. With these extensions, the
final efficiency achieved by our algorithm specifically among blocking
schemes is often substantially better and is found in a shorter time.

Beyond hierarchical clustering, there are other approaches one might
consider to find efficient blocking schemes.  One such approach
would be to use the structure of the graph instead of posterior
correlations to form blocks.  This would allow conservation of
calculations that are shared by some parts of the graph, whether or
not they are correlated.  Another future direction could be to improve
how a new kernel is determined from previous results, essentially to
determine an effective strategy for exploring the very
high-dimensional kernel space.  Finally, the trade-off
between computational cost and the accuracy of effective sample
size estimates is worth further exploration.

Development of Nested Adaptation MCMC has also raised several theoretical questions for future work.  First, similarly to other adaptive 
MCMC methods, can more general conditions of validity be proven?  Second, the estimation of effective sample size from a chain generated using 
Nested Adaptation is difficult because the properties of the kernel will have changed during the course of sampling.  Methods for estimating effective sample size typically don't consider such a situation.  Third, what are good choices for the external and internal adaptation schedules (e.g. $\gamma_n$ and $p_k$ in Section 2)?  Fourth, how can we disentangle the contribution of each sampler to mixing achieved by a kernel comprising multiple samplers?  Doing so could enable better moves in kernel space during outer adaptation.  Finally, what are good strategies for parallelization of Nested Adaptation?

\bibliographystyle{ba}
\bibliography{sample}

\begin{acknowledgement}
This work was funded in part by a grant from the U.S. National Science Foundation SI2-SSI program (ACI-1550488).
\end{acknowledgement}
\newpage
\appendix
\section{Block MCMC sampling}
This section provides the validity of block sampling in our general framework. 
We apply the results of \citet{latuszynski2013adaptive}
for block Gibbs samplers. Following closely their notation, let
$(\mathcal{X},\ B(\mathcal{X}))$ be an $m$-dimensional state space
where $\mathcal{X}=\mathcal{X}_{1}\times\cdots\times\mathcal{X}_{d}$,
$\mathrm{\mathcal{X}}_{i}\subseteq\mathbb{R}^{b_{i}}$ so that the
total dimension is $m=b_{1}+\cdots+b_{d}$. We write $X_{n}\in\mathcal{X}$
as $X_{n}=(X_{n,1},\ \ldots,\ X_{n,d})$ and set $\mathcal{X}_{-i}=\mathcal{X}_{1}\times\cdots\times\mathcal{X}_{i-1}\times\mathcal{X}_{i+1}\times\cdots\times\mathcal{X}_{d}$
so that 
\[
X_{n,-i}:=(X_{n,1},\ \ldots,\ X_{n,i-1},\ X_{n,i+1},\ \ldots,\ X_{n,d}).
\]
For $X\sim\pi$, let the conditional distribution of $X_{i}|X_{-i}=x_{-i}\sim\pi(\cdot|x_{-i})$.
We now consider a class of adaptive block MCMC sampler algorithms
as follows.

ALGORITHM of Nested Adaptation Block MCMC.

(1) Set $(\iota_{n},\theta_{n}):=R_{n}(\ X_{n-1},\ \ldots,\ X_{0},\ \iota_{n-1},\theta_{n-1},\ \ldots,\ \iota_{0},\theta_{0})\in\bar{\Theta}$
where $R_{n}$ is an outer adaptation function to select $(\iota_{n},\theta_{n})$
in $\bar{\Theta}$.

(2) Choose block $i\in\{1,\ \ldots,\ d\}$ in that order.

(3) Draw $Y\sim Q_{X_{n-1,-i},\iota_{n-1,i},\theta_{n-1,i}}(X_{n-1,i},\ \cdot)$ where $Q$ is a proposal distribution.

(4) Set $\begin{cases}
X_{n}=\ (X_{n-1,1},\ \ldots,\ X_{n-1,i-1},\ Y,\ X_{n-1,i+1},\ \ldots,\ X_{n-1,d}) & \mathrm{with\mbox{ }probability\,} p\\
X_{n}=X_{n-1} & \mathrm{otherwise}
\end{cases}$

where $p={\displaystyle \min(}1,{\displaystyle \frac{\pi(Y|X_{n-1,-i})q_{X_{n-1,-i},\iota_{n-1,i},\theta_{n-1,i}}(Y,X_{n-1,i})}{\pi(X_{n-1}|X_{n-1,-i})q_{X_{n-1,-i},\iota_{n-1,i},\theta_{n-1,i}}(X_{n-1,i},Y)})}.$
\vspace{3mm}

To establish the results for Nested Adaptation Block MCMC, we need the following
assumption.

\begin{assumption}\label{one} For every $i\in\{1,\ \ldots,\ d\},\ x_{-i}\in\mathcal{X}_{-i}$
and $(\iota_{i},\theta_{i})\in\bar{\Theta}$, the transition kernel
$P_{x_{-i},\iota_{i},\theta_{i}}$ is uniformly ergodic. Moreover,
there exist $s_{i}>0$ and an integer $m_{i}$ s.t. for every $x_{-i}\in\mathcal{X}_{-i}$
and $(\iota_{i},\theta_{i})\in\bar{\Theta}$, there exists a probability
measure $\nu_{x_{-i},\iota_{i},\theta_{i}}$ on $(\mathcal{X}_{i},\ B(\mathcal{X}_{i}))$,
s.t. $P_{x_{-i},\iota_{i},\theta_{i}}^{m_{i}}(x_{i},\ \cdot)\geq s_{i}\nu_{x_{-i},\iota_{i},\theta_{i}}(\cdot)$
for every $x_{i}\in\mathcal{X}_{i}$. \end{assumption}

We have the following results, similar to Theorem 4.19 of \citet{latuszynski2013adaptive}.

\begin{theorem}\label{autoblock} Assume that Assumption~\ref{one}
holds. Then Nested Adaptation Block MCMC is ergodic. That is, $\Vert\pi_{n}(x_{0})-\pi\Vert_{\mathrm{T}\mathrm{V}}\rightarrow0$
as $n\rightarrow\infty$. Moreover, if  
\[
{\displaystyle \sup_{x_{0}}\sup_{x\in\mathcal{X}}\Vert P_{x_{-i},\iota_{i+1},\theta_{i+1}}(x_{i},\ \cdot)-P_{x_{-i},\iota_{i},\theta_{i}}(x_{i},\ \cdot)\Vert_{\mathrm{T}\mathrm{V}}\rightarrow0}
\]
in probability, then convergence of Nested Adaptation Block MCMC is also uniform
over all $x_{0},$ that is, 
\[
{\displaystyle \sup_{x_{0}}\Vert\pi_{n}(x_{0})-\pi\Vert_{\mathrm{T}\mathrm{V}}\rightarrow0}
\]
as $n\rightarrow\infty$. \end{theorem}
\begin{proof}
To prove, we check simultaneous uniform ergodicity and the diminishing
adaptation property of the algorithm and apply the
result from Theorem 1 of \citet{roberts2007coupling} to conclude. We proceed as in the proof of \citet{latuszynski2013adaptive}, Theorem 4.10 
to establish simultaneous uniform ergodicity. By Assumption 1 and Lemma 4.14 of \citet{latuszynski2013adaptive}, 
every adaptive Metropolis transition kernel for the $i\mathrm{th}$ block, $P_{x_{-i},\iota_{n},\theta_{n}}$,
has stationary distribution $\pi(\cdot|x_{-i})$ and is ${\displaystyle \left(\left( \left\lfloor\frac{\log(s_{i}/4)}{\log(1-s_{i})}\right\rfloor+2\right)m_{i},\ \frac{{s}_{i}^{2}} 8\right)}$-strongly
uniformly ergodic. Since each block is sampled at least once by our construction and there is only a finite dimension, the probability of selecting each block is positive and bounded away 
from $0$. In other word, there exists $\epsilon$ such that $p_i>\epsilon$. Observe also that the family of block MCMC $(p),\ p_i>\epsilon$,
is $(m',\ s')$-strongly uniformly ergodic for some $(m',\ s')$ by \citet{latuszynski2013adaptive} Lemma 4.15. Hence, the family of 
adaptive block Metropolis-within-Gibbs samplers indexed by $\iota_{n},\theta_{n} \in\bar\Theta$, is $(m_{*},\ s_{*})$-simultaneously
strongly uniformly ergodic with some $m_{*}$ and $s_{*}$ given as in \citep{roberts1998two}. Therefore we have that simultaneous
uniform ergodicity is satisfied. 
Diminishing adaptation follows from Proposition 1 and from our construction. 
Therefore we can conclude the result.
\end{proof}

\end{document}